\documentclass[aps,pra,notitlepage,preprintnumbers,11pt,tightenlines]{revtex4-1}
\usepackage{amsmath}
\usepackage{amsthm}
\usepackage{amssymb}
\usepackage{bbold}
\usepackage{enumerate}
\usepackage{graphicx}
\usepackage{hyperref}
\usepackage{xspace}
\usepackage{color}
\usepackage[caption=false]{subfig}

\newcommand{\nc}{\newcommand}
\nc{\rnc}{\renewcommand}
\nc{\todo}[1]{\textcolor{red}{todo: #1}}
\nc{\<}{\langle}
\rnc{\>}{\rangle}

\nc{\cS}{\mathcal{S}}

\newcommand{\bra}[1]{\langle #1 |}
\newcommand{\ket}[1]{| #1 \rangle}
\newcommand{\eps}{\epsilon}
\renewcommand{\AA}{\mathcal{A}}
\newcommand{\CC}{\mathbb{C}}
\newcommand{\bbC}{\mathbb{C}}
\newcommand{\bbR}{\mathbb{R}}
\newcommand{\HH}{\mathcal{H}}
\newcommand{\PP}{\mathbb{P}}
\renewcommand{\P}{\mathcal{P}}
\newcommand{\RR}{\mathbb{R}}
\newcommand{\UU}{\mathcal{U}}
\newcommand{\VV}{\mathcal{V}}
\newcommand{\WW}{\mathcal{W}}
\newcommand{\XX}{\mathcal{X}}
\newcommand{\pd}{\partial}

\newcommand{\ident}{\mathbb{1}}
\newcommand{\ot}{\otimes}
\newcommand{\be}{\begin{equation}}
\newcommand{\ee}{\end{equation}}
\newcommand{\ra}{\to}
\newcommand{\proj}[1]{\ket{#1} \bra{#1}}
\DeclareMathOperator{\ProdSym}{ProdSym}
\DeclareMathOperator{\Sep}{Sep}
\DeclareMathOperator{\Tr}{Tr}
\DeclareMathOperator{\rank}{rank}
\DeclareMathOperator{\poly}{poly}
\DeclareMathOperator{\LT}{LT}
\DeclareMathOperator{\conv}{conv}
\DeclareMathOperator{\E}{\mathbb{E}}
\DeclareMathOperator{\diag}{diag}
\newtheorem{theorem}{Theorem}
\newtheorem{definition}[theorem]{Definition}
\newtheorem{lemma}[theorem]{Lemma}
\newtheorem{corollary}[theorem]{Corollary}
\newtheorem{proposition}[theorem]{Proposition}

\newcommand{\thmref}[1]{Theorem~\ref{thm:#1}}
\nc\eq[1]{(\ref{eq:#1})}

\newcommand{\problemmacro}[1]{\texorpdfstring{\textsc{#1}}{#1}\xspace}
\newcommand{\threesat}{\problemmacro{$3$-Sat}}

\begin{document}
\title{An improved semidefinite programming hierarchy for testing entanglement}

\author{Aram W. Harrow}
\author{Anand Natarajan}
\author{Xiaodi Wu}
\affiliation{MIT Center for Theoretical Physics}

\begin{abstract}

  We present a stronger version of the Doherty-Parrilo-Spedalieri (DPS) hierarchy of approximations for the set of separable states.  Unlike DPS, our hierarchy converges exactly at a finite number of rounds for any fixed input dimension.  This yields an algorithm for separability testing which is singly exponential in dimension and polylogarithmic in accuracy.  Our analysis makes use of tools from algebraic geometry, but our algorithm is elementary and differs from DPS only by one simple additional collection of constraints.

\preprint{MIT-CTP/4587}
\end{abstract}

\maketitle

\section{Introduction}\label{sec:intro}

Entanglement is one of the key features that distinguishes quantum
information from classical information.  One particularly basic and
important problem in the theory of entanglement is to determine
whether a given mixed state $\rho$ is entangled or separable.  Via
standard techniques of convex optimization, this problem is roughly
equivalent to maximizing a linear function over the set of separable
states~\cite{gls:1993,liu:2007}.  Indeed, it has close relations with a variety of
problems, including estimating channel capacities, analyzing
two-prover proof systems, finding the ground-state energy in the
mean-field approximation, finding the least entangled pure state in a
subspace, etc. as well as problems not obviously related to quantum
mechanics such as planted clique, the unique games problem and
small-set expansion~\cite{HM13}.

However, there is no simple test for determining whether a state is
entangled.  Indeed not only are tests such as the PPT (positive
partial transpose) condition known to have arbitrarily large
error~\cite{BeigiS10}, but computational hardness results show that
any test implementable in time polynomial in the dimension must be
highly inaccurate, given the plausible assumption that 3-SAT requires
exponential time~\cite{HM13,GallNN11}.
These limitations indicate that separability tests cannot be as
efficient as, say, a test for correlation, or a calculation of the
largest eigenvalue of a matrix.  

The main open question is whether algorithms exist that match these
hardness results, or whether further hardness results can be found.
The two leading algorithmic frameworks are $\eps$-nets and
semidefinite programming (SDP) hierarchies.  There are two regimes in
which these come close to matching the known hardness results.  Let
$n$ denote the dimension of the states we examine.  Informally
speaking, the well-studied regimes are the constant-error regime,
where there are both algorithms and hardness results with time
$n^{\Theta(\log n)}$ (although important caveats exist, discussed below),
and the $1/\poly(n)$ regime, where the algorithms and hardness results
together suggest that the complexity is exponential in $n$.

In this paper we consider the regime of much lower error.
Specifically, if $\eps$ is the error allowed, we will focus on the
scaling of error with $\eps$ rather than $n$.  In other settings, such
as infinite translationally invariant Hamiltonians, it is possible for
the complexity to grow rapidly with $1/\eps$ even for fixed local
dimension~\cite{CPW14}.  Another example closer to the current
work is \cite{ItoKW12}, which showed that approximating quantum
interactive proofs to high accuracy (specifically with the bits of
precision polynomial in the message dimension) corresponds to the complexity
class $\mathsf{EXP}$ rather than $\mathsf{PSPACE}$. However, for
separability testing or for the corresponding complexity class $\mathsf{QMA(2)}$,
we will give evidence that the complexity does not increase when $\eps$ becomes
exponentially small in the dimension.\footnote{On the other hand, there is
evidence that there is a separation between $\mathsf{QMA}(2)$ with
constant error and with error scaling as an inverse polynomial in the
dimension.~\cite{pereszlenyi:2012}.}

Our main contribution is to describe a pair of classical algorithms
for the separability problem.  In the high-accuracy limit both run in
time $\exp(\poly(n))\poly\log(1/\eps)$.  One is based on quantifier
elimination~\cite{Basu96} and is simple, but does not appear to yield
new insights into the problem.  The second algorithm is based on an
SDP hierarchy due to Doherty, Parrilo and Spedalieri
(DPS)~\cite{dps:2003}.  Like DPS, our algorithm runs in time
$n^{O(k)}$ (or more precisely $\poly(\binom{n+k-1}{k})$) for what is
called the $k^{\text{th}}$ ``level'' of the hierarchy.  As $k$ is
increased our algorithm, like that of DPS, becomes more accurate.
Indeed, for any fixed value of $k$ our algorithm performs at least as
well as that of DPS.  However, unlike DPS, our hierarchy always
converges exactly in a finite number of steps, which we can upper
bound by $\exp(\poly(n))$.  Taking into account numerical error yields
an algorithm again running in time $\exp(\poly(n))\poly\log(1/\eps)$.
Thus our algorithm is, for the first time, a single SDP hierarchy
which matches or improves upon the best known performance of previous
algorithms at each scale of $\eps$.

The fact that our algorithm is a semidefinite program gives it further
advantages. One very useful property of semidefinite programs is
duality. In our algorithm, both the primal and dual problems have useful
interpretations in terms of quantum information. On the primal side,
our algorithm can be viewed as searching over symmetric mixed states over
an extended system obtained by adding copies of the individual subsystems. In this light, our convergence bounds can be viewed as
new monogamy relations: we show that if a state is symmetric under
exchange of subsystems and satisfies certain other conditions, then if
there are enough copies of each subsystem, then none of the subsystems
can be entangled with each other. On the dual side, every feasible point of the
dual is an entanglement witness operator. Indeed, our algorithm yields
a new class of
entanglement witnesses, as discussed in
Section~\ref{sec:witness}. Duality is also useful in practice, since a
feasible solution to the dual can certify the correctness of the
primal, and vice versa.

SDP hierarchies are also used for discrete optimization problems, such
as integer programming~\cite{barak:2014}.  In that case, it is known that the $n^{\text{th}}$ level of
most SDP hierarchies provides the exact answer to optimization
problems on $n$ bits (e.g.~see Lemma 2.2 of \cite{barak:2014}).  By
contrast, neither the DPS hierarchy nor the 
more general Sum-of-Squares SDP hierarchy will converge exactly at any
finite level for general objective functions~\cite{dps:2003}.  Our result can be seen
as a continuous analogue of the exact convergence achievable for
discrete optimization.  

The main idea of our algorithm is that entanglement testing can be
viewed as a convex optimization problem, and thus the solution should
obey the KKT (Karush-Kuhn-Tucker) conditions.  Thus we can WLOG add
these as constraints.  It was shown
in \cite{nie:2010} that for general polynomial optimization problems,
adding the KKT conditions yields an SDP hierarchy with finite
convergence. Moreover, the number of levels necessary for convergence
is a function only of the number of variables and the degrees of the
objective and constraint polynomials. However, the proof of
convergence presented in \cite{nie:2010} gives a very high bound on the number of levels (triply
exponential in $n$ or worse). In contrast, we obtain a bound in the
number of levels that is singly exponential in $n$. We use tools from
algebraic geometry (B\'{e}zout's and Bertini's Theorem) to show that
generically, adding
the KKT conditions reduces the feasible set of our optimization
problem to isolated points. Then, using tools from
computational algebra (Gr\"{o}bner bases), we show that low levels of
the SDP hierarchy can effectively search over this finite
set. Although we use genericity in the analysis, our algorithm works
for all inputs. 

While some of these techniques have been used to analyze SDP
hierarchies in the past, they have generally not been applied to the
problems arising in quantum information.  We hope that they find
future application to understanding entanglement witnesses, monogamy
of entanglement and related phenomena.


Our main contribution is an improved version of the DPS hierarchy
which we describe in Section~\ref{sec:results}.   It is always at
least as stringent as the DPS hierarchy, and in Theorem~\ref{thm:main}
we show that it outperforms DPS by converging exactly at a finite
level, depending on the input dimension. We also present numerical
evidence in Section~\ref{sec:numerics} that the improved hierarchy outperforms DPS even at the
lowest nontrivial level for systems of small dimension.
\section{Background}\label{sec:background}
\subsection{Separability testing}
This section introduces notation and reviews previous work on the complexity of the separability testing
problem.  Define $\Sep(n,k) := \conv \{\proj{\psi_1} \ot \cdots \ot
\proj{\psi_k} : \ket{\psi_1},\ldots,\ket{\psi_k} \in B(\bbC^n)\}$,
where $\conv(S)$ denotes the convex hull of a set $S$ (i.e. the set of
all finite convex combinations of elements of $S$) and $B(V)$ denotes
the set of unit vectors in a vector space $V$.  States in $\Sep(n,k)$
are called separable, and those not in $\Sep(n,k)$ are entangled.
Given a Hermitian matrix
$M$, we define
\be h_{\Sep(n,k)}(M) := \max\{ \Tr[M\rho] : \rho\in \Sep(n,k)\}.
\label{eq:hSep}\ee
We will often abbreviate
$\Sep := \Sep(n,2)$ where there is no ambiguity.
More generally if $K$ is a convex set, we can define $h_K(x) :
=\max\{\langle x,y\rangle : y \in K \}$.  

A classic result in convex optimization~\cite{gls:1993} holds that
approximating $h_K$ is roughly equivalent in difficulty to the weak
membership problem for $K$: namely, determining whether $x\in K$ or
whether $\text{dist}(x,K)>\eps$ given the promise that one of these
holds.  This was strengthened in the context of the set $\Sep$ by
Gharibian~\cite{gharibian10} to show that this equivalence holds when
$\eps \leq 1/\poly(n)$.  Thus, in what follows we will treat
entanglement testing (i.e. the weak membership problem for $\Sep$) as
equivalent to the optimization problem in \eqref{eq:hSep}.

\subsubsection{Related problems}\label{sec:related}
A large number of other optimization problems are also equivalent
to $h_{\Sep}$, or closely related in difficulty.  Many of these are
surveyed in \cite{HM13}.  One that will particularly useful will be
the optimization problem $h_{\ProdSym(n,k)}$, defined in terms of the set
$\ProdSym(n,k) := \conv\{(\proj{\psi})^{\ot k} : \ket\psi\in
B(\bbC^n)\}$.  In Corollary 14 of \cite{HM13} (see specifically explanation
(2) there) it was proven that for any $n^2$-dimensional $M$ there exists
$M'$ with dimension $4n^2$ satisfying
\be  h_{\ProdSym(2n,2)}(M') = \frac{1}{4} h_{\Sep(n,2)}(M).\ee
Thus an algorithm for $h_{\ProdSym}$ implies an algorithm of similar
complexity for $h_{\Sep}$.   In the body of our paper, we will
describe an algorithm for the mathematically simpler $h_{\ProdSym}$,
with the understanding that it also covers the more widely used
$h_{\Sep}$.

We will not fully survey the applications of separability testing, but
briefly mention two connections.  First, $h_{\Sep(2^n,k)}$ is closely related
to the complexity class $\mathsf{QMA_n(k)}$ in which $k$ unentangled provers
send $n$-qubit states to a verifier.  If the verifier's measurement is
$M$ (which might be restricted, e.g. by being the result of a short
quantum circuit) then the maximum acceptance probability is precisely
$h_{\Sep(2^n,k)}(M)$.  Thus the complexity of $h_{\Sep}$ is closely
related to the complexity of multiple-Merlin proof systems.  See
\cite{AMM} for a classical analogue of these proof systems, and a
survey of recent open questions.

Second, $h_{\Sep}$ is closely related to the  problems of estimating the $2\ra 4$ norm
of a matrix, finding the least-expanding small set in a graph and
estimating the optimum value of a unique game~\cite{BBHKSZ12}. These
problems in turn relate to the approximation complexity of constraint
satisfaction problems, which are an extremely general class of
discrete optimization problems.  They are
currently known only to be of intermediate complexity
(i.e. only subexponential-time algorithms are known), and are the subject of intense
research.  One of leading approaches to these problems has been SDP
hierarchies, but here too it is generally unknown how well these
hierarchies perform or which features are important to their success.

\subsubsection{Previous algorithms and hardness results}\label{sec:previous-sep}
Algorithms and hardness results for estimating $h_{\Sep(n,2)}(M)$ can
be classified by (a) the approximation error $\eps$, and (b)
assumptions (if any) for the matrix $M$.  In what follows we will
assume always that $0 \leq M \leq I$.  Define $\threesat[m]$ to be the
problem of solving a 3-SAT instance with $m$ variables and $O(m)$
clauses.  The exponential-time hypothesis (ETH)~\cite{ImpagliazzoPZ98}
posits that $\threesat[m]$ requires time $2^{\Omega(m)}$ to solve.

The first group of hardness
results~\cite{gurvits:2003,BT09,BeigiS10,ChiesaF13,GallNN11} for $h_{\Sep(n,2)}$
have $\eps \sim 1/\poly(n)$ and yield reductions from $\threesat[n]$.
The strongest of these results~\cite{GallNN11} achieves this with $\eps
\sim 1/n\poly\log(n)$.  As discussed above, there are algorithms that
come close to matching this.  Taking $k=n/\sqrt{\eps}$ in the DPS
hierarchy achieves error $\eps$ (see \cite{NOP09}) in time
$(n/\sqrt{\eps})^{O(n)}$, which is $n^{O(n)}$ when $\eps =
1/\poly(n)$.  An even simpler algorithm is to enumerate over an
$\eps$-net over the pure product states on $\bbC^n \ot \bbC^n$.  Such
a net has size $(1/\eps)^{O(n)}$, which again would yield a run-time
of $n^{O(n)}$ if $\eps = 1/\poly(n)$.
Thus neither algorithm nor the hardness result could
be significantly improved without violating the ETH.  However, the
value of $\eps$ in the hardness result could conceivably be reduced.

The second body of work has concerned the case when $\eps$ is a
constant.  Here the existing evidence points to a much lower complexity.
Constant-error approximations for $h_{\Sep(n,\sqrt{n}\poly\log(n))}(M)$
were shown to be as hard as $\threesat[n]$ in \cite{AaronsonBDFS08}
and in \cite{chen10} this was shown to still hold when $M$ is a Bell
measurement (i.e. each system is independently measured and the
answers are then classically processed).  This was extended to
bipartite separability in \cite{HM13} which showed the $\threesat[n]$-hardness of
approximating 
$h_{\Sep(\exp(\sqrt{n}\poly\log(n)),2)}(M)$ to constant accuracy.
There it was shown that $M$ could be taken to be separable (i.e. of
the form $\sum_i A_i \ot B_i$ with $A_i,B_i \geq 0$) without loss of
generality.  Scaling down this means that $h_{\Sep(n,2)}$ requires
time $n^{\tilde\Omega(\log(n))}$ assuming the ETH.  On the algorithms side, $O(\log(n)/\eps^2)$ levels of the
DPS hierarchy are known~\cite{brandao11,LW14,BH-local} to suffice when
$M$ is a 1-LOCC measurement (i.e. separable with the extra assumption
that $\sum_i A_i \leq I$).   This also yields a runtime of
$n^{O(\log(n)/\eps^2)}$, but does not match the hardness result of
\cite{HM13} because of the 1-LOCC assumption.  Similar results are
also achievable using $\eps$-nets~\cite{shi12,BH14}.  One setting
where the hardness result is known to be tight is when there are many
provers.  When $M$ is implemented by $k-1$ parties measuring locally
and sending a message to the final party, \cite{BH-local} showed that
DPS could approximate the value of $h_{\Sep(n,k)}(M)$ in time
$\exp(k^2\log^2(n)/\eps^2)$.  This nearly matches the hardness result
of \cite{chen10} described above.  The same runtime was recently shown to work
for a larger class of $M$ in \cite{LS14}.

\subsection{Sum-of-squares hierarchies}\label{sec:sos}

Here we review the general method of sum-of-squares relaxations for
polynomial optimization problems. In this section, all variables are
real and all polynomials have real coefficients, unless otherwise
stated.  To start with, let $g_1(x), \dots g_k(x)$ be polynomials in
$n$ variables and define
$V(I) = \{ x \in \RR^n: \forall_i g_i(x) = 0\}$.  This notation
reflects the fact that $V(I)$ is the variety corresponds to the ideal
$I$ generated by $g_1(x),\ldots,g_m(x)$; see
Appendix~\ref{subsec:alg-geom-basics} for definitions and more
background on algebraic geometry.

Now given another polynomial $f(x)$, suppose we would like to prove
that $f(x)$ is nonnegative for all $x\in V(I)$. One way to do this would be write
$f$ as
 \be f(x) = \sum_j a_j(x)^2 + \sum_i b_i(x) g_i(x), \ee 
for
polynomials $\{a_j(x)\}, \{b_i(x)\}$.  The first term on the RHS is a
sum of squares, and is thus non-negative everywhere, while the second
term is zero everywhere on $V(I)$. Thus, if such a decomposition for
$f(x)$ exists, it must be nonnegative on $V(I)$. Such a decomposition is
thus called a \emph{sum-of-squares (SOS) certificate} for the
nonnegativity of $f$ on $V(I)$.

A natural question to ask is whether all nonnegative polynomials on
$S$ have a SOS certificate. A positive answer to this question is
provided under certain conditions by Putinar's Positivstellensatz
\cite{putinar:1993}.   One such condition is the 
 {\em Archimedean condition}, which asserts that there exists a constant
$R>0$ and a sum-of-squares polynomial $s(x)$ such that
\be R - \sum_i x_i^2 - s(x) \in I .\ee
Equivalently we could say that there is a SOS proof of $x\in V(I)
\Rightarrow \sum_i x_i^2
\leq R$.
This condition generally holds whenever $V(I)$ is a
manifestly compact set.  In this case, we have the following
formulation of Putinar's Positivstellensatz from Theorem A.4 of \cite{nie:2010}.
\begin{theorem}[Putinar]
  Let $I$ be a polynomial ideal satisfying the Archimedean condition
  and $f(x)$ a
  polynomial with $f(x) > 0$ for all $x \in V(I) \cap \RR^n$.  Then there
  exists a sum-of-squares polynomial $\sigma(x)$ and a real
  polynomial $g(x) \in I$ such that
  \[ f(x) = \sigma(x) + g(x). \]
  \label{thm:psatz}
\end{theorem}
Neither Putinar's Positivstellensatz, nor the Archimedean condition, put any bound on the
degree of the SOS certificate. Now suppose we would like to solve a general polynomial optimization
problem:
\be \begin{aligned}
  \label{eq:polyopt}
  &\max && f(x) \\
  &\text{subject to} && g_i(x) = 0 \; \forall i
  \end{aligned}. \ee
 We can rewrite this in terms of polynomial positivity as follows:
\be \begin{aligned}
  &\min && \nu \\
  &\text{such that} &&\nu - f(x) \geq 0 \\
  &\text{whenever} &&g_i(x) = 0 \; \forall i
\end{aligned}. \ee
Now, if the ideal $\langle \{g_i(x)\} \rangle$ generated by the
constraints obeys the Archimedean condition, then
Putinar's Positivstellensatz means that this problem is equivalent to
\be \label{eq:sos_program} \begin{aligned}
  &\min && \nu \\
  &\text{ such that} && \nu - f(x) = \sigma(x) + \sum_{i} b_i(x)
  g_i(x),
  \end{aligned} 
  \ee
  where $\sigma(x)$ is SOS and the polynomials $b_i(x)$ are arbitrary.
  If we allow $\sigma(x)$ and $b_i(x)$ to have arbitrarily high
  degrees, then the problem in this form is exactly equivalent to the
  original problem, but it involves optimizing over an
  infinite number of variables. However, if we limit the degrees, so
  that $\deg(\sigma(x)), \deg(b_i(x) g_i(x)) \leq 2D$ for some integer
  $D$, then we obtain a problem over a finite number of variables. As we
  increase $D$, we get a hierarchy of optimization problems over
  increasingly more variables, which must converge to the
  original problem. 

  It remains to show how to perform the optimization over a degree-$2D$
  sum of squares certificate. It turns out that this optimization can
  be expressed as a \emph{semidefinite program}. The idea is that any
  polynomial $g(x)$ of degree $2D$ can be represented as a quadratic
  form $m^T Q m$, where $m$ is the vector of monomials of degree up to
  $2D$. Moreover, the polynomial $g(x)$ is SOS iff the matrix $Q$ of
  the corresponding quadratic form is positive semidefinite.   One
  direction of this equivalence is as follows.
  If $g(x) = \sum_i h_i(x)^2$ then each $h_i(x) = \langle \vec h_i,
  m\rangle$ for some vector $\vec h_i$, and we have $Q = \sum_i \vec
  h_i \vec h_i^T$.    The reverse direction follows from the fact that
  any psd $Q$ can be decomposed in this way.

  The SDP associated with the optimization in \eqref{eq:sos_program}
  is \be \begin{aligned}
    &\min_{\nu, b_{i\alpha} \in \RR} && \nu \\
    &\text{such that} && \nu A_0 - F - \sum_{i \alpha} b_{i\alpha}
    G_{i\alpha} \succeq 0.
    \end{aligned} \ee
    Here $A_0$ is the matrix corresponding to the constant polynomial
    $1$, $F$ is the matrix corresponding to $f(x)$, $\alpha$ is a
    multi-index labeling monomials, and $G_{i\alpha}$ is the matrix
    representing the polynomial $x_1^{\alpha_1} \dots x_n^{\alpha_n}
    g_i(x)$. These matrices have dimension $m
    \times m$, where $m$ is the number of monomials of degree at most
    $D$. For $n$ variables, $m = \binom{n + D}{D}$. There exist
    efficient algorithms to solve SDPs: if desired numerical
    precision is $\epsilon$, and all feasible solutions have norm
    bounded by a constant $R$, then the running time for an SDP over
    $m \times m$ matrices is $O(\poly(m) \poly\log(R/\epsilon))$. For
    a more detailed discussion of SDP complexity, see e.g. \cite{gls:1993}.

These general techniques were applied to the separability testing
problem by Doherty, Parrilo and Spedalieri in \cite{dps:2003}.
We refer to resulting SDP as the DPS relaxation.  For a state
$\rho^{AB}$, the level-$k$ DPS relaxation asks whether there exists an
extension $\tilde\rho^{A_1 \ldots A_k B_1 \ldots B_k}$ invariant under
left or right-multiplying by any permutation of the $A$ or $B$
systems and that remains PSD under transposing any subset of the
systems.  This latter condition is called Positivity under Partial
Tranpose (PPT).  It is straightforward to see that searching for such
a $\tilde\rho$ can be achieved by an SDP of size $n^{O(k)}$.  In
\cite{NOP09} it was proven that the level-$k$ DPS relaxation produces
states within trace distance $O(n^2/k^2)$ of the set of separable
states.  Of course this bound is vacuous for $k < n$, but limited
results are known in this case as well; cf. the discussion in~
\ref{sec:previous-sep}.
 
Often weaker forms of DPS are analyzed.  For example, we might demand
only that an extension of the form $\tilde \rho^{AB_1\ldots B_k}$
exist, or might drop the PPT condition.  Many proof techniques
(e.g. those in \cite{brandao11} and followup papers) do not take
advantage of the PPT condition, for example, although it is known that
without it the power of the DPS relaxation will be limited (see
e.g. \cite{BuhrmanRSW11}).  Our approach will be to instead {\em add}
constraints to DPS.

\section{Results}\label{sec:results}
\subsection{Separability as polynomial optimization}

As discussed in Section~\ref{sec:related}, a number of problems in entanglement can be
reduced to the problem $h_{\ProdSym(n,d)}$: 
\begin{equation}
  \label{eq:hprodsym_density}
  \begin{aligned}
    &\max_{\rho \in \ProdSym(n,d)} && \Tr[M \rho].
  \end{aligned} 
\end{equation}
Since $\ProdSym(n,d)$ is a convex set, the maximum will be attained on
the boundary, which is the set of pure product states $\rho = (\ket{a}
\bra{a})^{ \otimes d}$. We can rephrase the optimization in terms of
the components of this pure product state.
\begin{equation}
\label{eq:hprodsym_components}
\begin{aligned}
  &\max_{a \in  \CC^n} && \sum_{i_1 \dots i_k j_1 \dots j_d} M_{(i_1
    \dots i_d), (j_1 \dots j_d)} a_{i_1}^* \dots a_{i_d}^*  a_{j_1} \dots a_{j_d}  \\
  &\text{subject to} && ||a||^2 =1.
\end{aligned}
\end{equation}
This is an optimization problem over the complex vector space $\CC^n$. We can
convert it to a real optimization problem over $\RR^{2n}$ by explicitly decomposing
the complex vectors into real and imaginary parts. Since the matrix
$M$ is hermitian, the objective function in
\eqref{eq:hprodsym_components} is a real polynomial in the real and imaginary parts
of $a$. Thus, we can write the problem as
\begin{equation}
  \label{eq:hprodsym_complex}
  \begin{aligned}
  &\max_{x \in \RR^{2n}} & & \sum_{i_1 \dots i_d j_1 \dots j_d} \tilde{M}_{(i_1 \dots i_d),( j_1 \dots j_d)}
  x_{i_1} \dots x_{i_d} x_{j_1} \dots x_{j_d} \\
  &\text{subject to} & & ||x||^2 -1 = 0 \\
\end{aligned}
\end{equation}
We will denote this problem by $h_{\ProdSym(\RR, 2n, d)}(\tilde{M})$. Here the matrix $\tilde{M}$ has dimension $(2n)^d \times
(2n)^d$. We can alternatively view $\tilde{M}$ as an object with $2d$
indices, each of which ranges from $1$ to $2n$. We call this a
\emph{tensor of rank $2d$}. Without loss of generality, we
can assume that $\tilde{M}$ is completely symmetric under all permutations of
the indices. Henceforth, we will only work with real variables, so we
will drop the tilde and just write $M$. For compactness' sake we will use the
notation $\langle M,x^{\ot 2d}\rangle$ to mean the contraction of $M$, viewed as a rank
$2d$ tensor, with $2d$ copies of the vector $x$. In this notation, the problem
$h_{\ProdSym(\RR, 2n, d)}(M)$ becomes:
\begin{equation}
  \label{eq:hprodsym}
  \begin{aligned}
  &\max_{x \in \RR^n} & & f_0(x) \equiv \< M,  x^{\ot 2d}\> \\
  &\text{subject to} & & f_1(x) \equiv ||x||^2 -1 = 0 \\
\end{aligned}
\end{equation}

Our first algorithm for this problem uses quantifier
elimination~\cite{Basu96} to solve \eqref{eq:hprodsym} in a black-box
fashion.  This yields an algorithm with runtime
$d^{O(n)}\poly\log(1/\eps)$.
\begin{theorem}\label{thm:quant-elim}
There exists an algorithm to estimate \eqref{eq:hprodsym} to
multiplicative accuracy $\eps$
in time $d^{O(n)}\poly\log(1/\eps)$.
\end{theorem}
Estimating a number $X$ to multiplicative accuracy $\eps$ means
producing an estimate $\hat X$ satisfying $|X - \hat X| \leq \eps
|X|$, while additive accuracy $\eps$ means that $|X-\hat X| \leq \eps$.

\begin{proof}
Assume WLOG that $M$ is supported on the symmetric subspace and 
has been rescaled such that $\|M\| = 1$.   Then
\be
h_{\ProdSym(n,d)}(M) \geq
\E_{\ket a} \Tr[M \proj{a}^{\ot d}]
= \frac{\Tr[M]}{\binom{n+d-1}{d}}
\geq \|M\| n^{-d}.\ee
Thus it will suffice to achieve additive error $\eps' := \eps/n^d$.

Theorem 1.3.3 of \cite{Basu96} states that polynomial equations of the
form 
\be \exists x\in \RR^n, g_1(x)\geq 0, \ldots, g_m(x)\geq 0
\label{eq:existential}\ee
 can be
solved using $(md)^{O(n)}$ arithmetic operations.   Moreover if the
$g_1,\ldots,g_m$ have integer coefficients with absolute value $\leq
L$ then the intermediate numbers during this calculation are integers
with absolute value $\leq L (md)^{O(n)}$.
We can put \eqref{eq:hprodsym} into the form \eqref{eq:existential} (with $m = O(1)$) by adding a
constraint of the form $f_0(x) \geq \theta$ and then performing binary
search on $\theta$, starting with the {\em a priori} bounds $0 \leq
h_{\ProdSym}(M) \leq \|M\| \leq 1 $.   If we specify the entries of $M$ to
precision $\eps'/ \poly(n)$ then this will induce operator-norm error
$\leq \eps'$, which implies error $\leq \eps'$ in $h_{\ProdSym}$.  Thus
we can take $L \leq \poly(n)/\eps' \leq n^{d+O(1)}/\eps$.  Since
arithmetic operations on numbers $\leq L$ require $\poly\log(L)$ time,
we attain the stated run-time.
\end{proof}

The advantage of this argument is that it is simple and yields an
effective algorithm.  However, SDP hierarchies have several advantages
over \thmref{quant-elim}.  The dual of an SDP can be useful, and here
corresponds to entanglement witnesses, as we discuss in
\ref{sec:witness}.  An SDP hierarchy can interpolate in runtime between
polynomial and exponential, whereas the algorithm in
\thmref{quant-elim} can only be run in exponential time.  Finally the
hierarchy we develop can be interpreted in terms of extensions of
quantum states and therefore has an interpretation in terms of a
monogamy relation, although developing this is something we leave for
future work.

We now turn towards developing an improved SDP hierarchy for approximating
$h_{\ProdSym}$ in a way that will be at least as good at DPS at the
low end and will match the performance of \thmref{quant-elim} at the
high end.
The objective function and constraints in \eqref{eq:hprodsym} are both smooth, so the
maximizing point must satisfy the \emph{Karush-Kuhn-Tucker
  (KKT) conditions}:
\[ \rank \begin{pmatrix} \frac{\pd f_0(x)}{\pd x_1} & \frac{\pd
    f_1(x)}{\pd x_1} \\ \vdots & \vdots \\ \frac{\pd f_0(x)}{\pd x_{2n}}
  & \frac{\pd f_1(x)}{\pd x_{2n}} \end{pmatrix} < 2. \]
This rank condition is equivalent to the condition that all $2\times
2$ minors of the matrix should be equal to zero. Each minor is a
polynomial of the form
\be g_{ij}(x) = \frac{\pd f_0(x)}{\pd x_i} \frac{\pd f_1(x)}{\pd x_j} -
\frac{\pd f_0(x)}{\pd x_j} \frac{\pd f_1(x)}{\pd x_i}.
\label{eq:gij-def} \ee
Note that $\deg(g_{ij}(x)) = \deg(\<M, x^{\ot 2d}\>) = 2d$. If we add these conditions to \eqref{eq:hprodsym}, we get the
following equivalent optimization problem:
\begin{equation}
  \label{eq:hsep_kkt}
  \begin{aligned}
    &\max_{x \in \RR^{2n}} & & f_0(x) \\
    &\text{subject to} & & f_1(x) =0\\
    &&& g_{ij}(x) = 0 &  \forall \, 1 \leq i, j \leq 2n
  \end{aligned}
\end{equation}

\subsection{Constructing the Relaxations}
We will now construct SDP relaxations for this problem.
Our first step will be to express \eqref{eq:hsep_kkt} in terms of polynomial
positivity:
\[
\begin{aligned}
  &\min && \nu \\
  &\text{such that} &&  \nu 
\<\ident^{\otimes d}, x^{\ot 2d}\>- f_0(x) \geq 0\\
  &\text{whenever} && f_1(x)
  = 0 \\
  &&& g_{ij}(x) = 0 &
  \forall \, 1 \leq i, j \leq 2n
\end{aligned}
\]
Here $\ident$ is the identity matrix.
Note that we have multiplied $\nu$ by $\<\ident^{\otimes d}, x^{\ot 2d}\> =
||x||^{2d} $; we are free to do this because
this factor is equal to $1$ whenever the norm constraint is satisfied. 
Now, as we described in Section~\ref{sec:sos}, we replace the positivity constraint with
the existence of an SOS certificate.
\begin{equation}
\label{eq:hsep_kkt_sos}
\begin{aligned}
  &\min && \nu \\
  &\text{such that} && \<\nu \ident^{\otimes d} -M, x^{\ot 2d}\>=
  \sigma(x) + \phi(x) f_1(x) + \sum_{ij} \chi_{ij}(x) g_{ij}(x)\\
\end{aligned}
\end{equation}
Here $\sigma$ is a sum of squares and $\phi, \chi_{ij}$ are arbitrary
polynomials. We can now produce a hierarchy of relaxations by
varying the degree of the certificates $(\sigma, \phi, \chi_{ij})$ that we
search over. Specifically, at the $r$th level of the hierarchy, the
total degree of all terms in the SOS certificate is upper-bounded by
$2(r + d)$.

\subsubsection{Explicit SDPs}
The formulation \eqref{eq:hsep_kkt_sos}  of the hierarchy in terms of
SOS polynomials will be the one we use for most of our analysis. However, there is an
alternative formulation in terms of an explicit SDP over moment matrices, which is more
convenient for some purposes. Before we derive it, we will first make
some simplifications that will let us eliminate the polynomial
$\phi(x)$. Suppose we are working at level $r$ of the hierarchy, so
all the terms in the certificate have degree at most $2(d+r)$. Without loss of generality, we can
assume that all terms in $\phi(x)$ and $\chi_{ij}(x)$ have even
degree~\footnote{The LHS of \eqref{eq:hsep_kkt_sos} only contains even
  degree terms, as do $f_1(x)$ and $g_{ij}(x)$. Thus, any odd degree
  terms in $\phi(x)$ and $\chi_{ij}(x)$ must cancel each other, so
  they can all be removed.}. Moreover, we claim that without loss of generality, all the
polynomials $\chi_{ij}$ are homogeneous of degree $2r$. Indeed,
suppose $\chi_{ij}$ contains a term $a$ of degree $2(r - k)$. Then $a
= ||x||^{2k} a + (1 - ||x||^{2k})a$. Since $f_1(x) = ||x||^2 - 1$
divides $||x||^{2k} - 1$ for all $k \geq 1$, this means we can replace
$a$ with $||x||^{2k} a$ and absorb the error term inside
$\phi(x)$.

Now we can eliminate $\phi(x)$ using the following
argument, which is based on Proposition 2
in \cite{deklerk:2005}. Denote the LHS of
\eqref{eq:hsep_kkt_sos} by $q(x)$ and observe that it is homogeneous of degree
$2d$. 
Then since $f_1(x/\|x\|)=0$ we have
\begin{align*}
  q\left( \frac{x}{||x||}\right) &= \sigma\left(\frac{x}{||x||}\right)
  + \sum_{ij} \chi_{ij}\left(\frac{x}{||x||}\right)
  g_{ij}\left(\frac{x}{||x||}\right) \\
  q(x) ||x||^{2r} &= \sigma\left(\frac{x}{||x||}\right) ||x||^{2(r+d)} +
  \sum_{ij} \chi_{ij}(x) g_{ij}(x)
\end{align*}
Since $\sigma$ has degree at most $2(r+d)$, 
$\sigma'(x) \equiv \sigma\left(\frac{x}{||x||}\right) ||x||^{2(r+d)}$ is a polynomial
in $x$. Moreover, by expanding the $\sigma(x) = \sum_k a_k(x)^2$, one
can check that $\sigma'(x) = \sum_a s_{a}^{2}(x)$ where each term $s_a$ is homogeneous
of degree $r + d$. We say that $\sigma'(x)$ is a sum of homogeneous
squares. Thus, from a certificate of the form given in \eqref{eq:hsep_kkt_sos}, we have constructed a new
certificate of the form
\begin{equation}
  q(x)||x||^{2r}  = \sigma'(x) + \sum_{ij} \chi_{ij}(x) g_{ij}(x),
  \label{eq:polya_cert}
\end{equation}
with $\sigma'$ a sum of homogeneous squares. In this form we have
eliminated the polynomial $\phi$. Conversely, from any certificate of the form
\eqref{eq:polya_cert}, we can produce a certificate in the form \eqref{eq:hsep_kkt_sos}
as follows:
\begin{align*}
  q(x)||x||^{2r}  &= \sigma'(x) + \sum_{ij} \chi_{ij}(x) g_{ij}(x)
  \\
  q(x) &= \sigma'(x) + q(x)(1 - ||x||^{2r}) + \sum_{ij} \chi_{ij}(x)
  g_{ij}(x)
\end{align*}
Since $1 - ||x||^2$ divides $1 - ||x||^{2r}$, this is indeed a
certificate of the form given in \eqref{eq:hsep_kkt_sos}.
Thus, we have shown that the hierarchy \eqref{eq:hsep_kkt_sos} is
equivalent to the following hierarchy.
\begin{equation}
  \label{eq:hsep_kkt_polya}
  \begin{aligned} 
  &\min && \nu \\
  &\text{such that} && 
\< \nu \ident^{\otimes(d+r)} -M \otimes
  \ident^{\otimes r}, x^{\ot 2(d+r)}\>- \sum_{ij}
  \chi_{ij}(x) g_{ij}(x) = \sigma(x) . \\
\end{aligned}
\end{equation}
Here, $\chi_{ij}(x)$ is an arbitrary homogeneous polynomial of degree
$2r$ and $\sigma(x)$ is a sum-of-homogeneous-squares polynomial of
degree $2(d+r)$. 

This SOS program can be written explicitly as an SDP, using the
procedure described in Section~\ref{sec:sos}. This would produce an
SDP over $m \times m$ matrices where $m = \binom{2n + 2(d+r) -
  1}{2(d+r)}$ is the number of monomials of degree
$2(d+r)$. This SDP can be solved to accuracy $\epsilon$ in time
$O(\poly(m) \poly\log(1/\epsilon))$. 

However, in order to facilitate comparison with DPS, we will
instead write an SDP over $(2n)^{d+r} \times (2n)^{d+r}$
matrices; this corresponds to treating different orderings of the
variables in a monomial as distinct monomials. The redundant degrees
of freedom will be removed by imposing symmetry
constraints. 
Specifically, let the map $\P$ from tensors of rank
$2k$ to matrices of dimension $(2n)^{k}$ be defined by
\[ (\P A)_{(i_1 i_2 \dots i_k), (i_{k+1} i_{k+2} \dots i_{2k})} \equiv \frac{1}{(2k)!} \sum_{\pi \in
  \mathcal{S}_{2k}} A_{i_{\pi(1)} i_{\pi(2)} \dots
  i_{\pi(2k)} } , \]
where $\mathcal{S}_{2k}$ is the group of all permutations of $\{1, \dots, 2k\}$.
Then our SDP is
\begin{equation}
\label{eq:hsep_kkt_polya_sdp}
\begin{aligned}
  &\min && \nu \\
  &\text{such that} && \P \left( \nu \ident^{d+r} - M \otimes
    \ident^{\otimes r} - \sum_{ij\alpha} \chi_{ij\alpha} A_\alpha \otimes \Gamma_{ij}
  \right) \succeq 0.
\end{aligned}
\end{equation}
Here, the indices $ij$ label the KKT constraints, and
the multi-index $\alpha$ labels all monomials of degree $2r$. The
variable $\chi_{ij\alpha}$ is the coefficient of the monomial
$\alpha$ in the polynomial $\chi_{ij}$. The matrix $A_\alpha$
represents the monomial $\alpha$, i.e. $\<A_\alpha, x^{\ot 2r}\> =
x_1^{\alpha_1} \dots x_n^{\alpha_n}$. Finally, the matrix
$\Gamma_{ij}$ represents the KKT polynomial $g_{ij}(x)$,
i.e. $\<\Gamma_{ij}, x^{\ot 2d} \>= g_{ij}(x)$.

Now we can at last write down the moment matrix version of the
hierarchy by applying SDP duality to \eqref{eq:hsep_kkt_polya_sdp}.
\begin{equation}
  \label{eq:hsep_kkt_moment}
  \begin{aligned}
    &\max_{\rho} && \langle \P(M \otimes \ident^{\otimes r}), \rho \rangle \\
    &\text{such that} && \rho \succeq 0 \\
    &&& \langle \P(A_\alpha \otimes \Gamma_{ij} ), \rho \rangle = 0
    \; \forall i, j, \alpha.
  \end{aligned}
\end{equation}
In this program, the variable $\rho$ is a matrix in $\RR^{(2n)^{d+r} \times (2n)^{d+r}}$. Now we see the advantage of adding the redundant degrees of
freedom in the SDP---just as in DPS, $\rho$ can be interpreted as the density matrix
over an extended quantum system. The main difference from DPS
is the set of added constraints $\langle A_\alpha \otimes \Gamma_{ij}, \rho
\rangle = 0$, which are the moment relaxations of the KKT
conditions.

The SDP~\eqref{eq:hsep_kkt_moment} is over $(2n)^{d+r} \times
(2n)^{d+r}$ matrices, so if $r = O(\exp(n))$, we would na\"{i}vely
expect it to have time complexity $O(\exp(\exp(n) \log(n)))$. This
apparently large complexity is caused by the redundant degrees of
freedom we added above. In practice, we can use the symmetry constraints
enforced by $\P$ to eliminate the redundancy and bring the complexity back down to
$\binom{2n + 2(d+r) - 1}{2(d+r)}^{O(1)}$, which is $O(\exp(n))$ when
$r = O(\exp(n))$. This is discussed in more detail in Section IV of
the original DPS paper~\cite{dps:2003}.

\subsection{Degree bounds for SOS certificates}
In this section, we will show that for generic inputs, the SOS form of
the hierarchy
\eqref{eq:hsep_kkt_sos} converges exactly within $d^{O(n^2)}$
levels. In other words, we will show that generically, there exists a
sum-of-squares certificate of degree $O(d^{\poly(n)})$.  This is an algebraic statement, so it is useful to recast it in the
language of polynomial ideals. We define the \emph{KKT ideal} $I_K$ to
be the ideal generated by the polynomials $g_{ij}$ and $f_1$. Likewise,
define the \emph{truncated} KKT ideal $I_K^m$ to be
\[I_K^m = \left\{ v(x) f_1(x) + \sum_{ij} h_{ij}(x) g_{ij}(x) :
  \deg(v(x) f_1(x)) \leq m, \max_{i, j} \deg(h_{ij}(x) g_{ij}(x)) \leq m \right\}. \]

Then we claim
\begin{theorem}\label{thm:main}
  Let $f_0, f_1, g_{ij}$ be as defined in \eqref{eq:hprodsym}. Then there
  exists $m = d^{O(n^2)}$ such that for generic $M$,
  if $\nu - f_0(x) > 0$ for all $x \in \mathbb{R}^{2n}$ such that
  $f_1(x) = 0$, then 
  \[ \nu - f_0(x) = \sigma(x) + g(x), \]
  where $\sigma(x)$ is sum of squares, $\deg(\sigma(x)) \leq m$ and
  $g(x) \in I_K^m$.
  \label{thm:sos_cert}
\end{theorem}

The proof is in Section~\ref{sec:proofs}.

\begin{corollary}
We can estimate $h_{\ProdSym(n,2)}$ to multiplicative error $\eps$ in time $\exp(\poly(n))\poly\log(1/\eps)$.
\end{corollary}

This follows from \thmref{main} and the fact that the value of semidefinite
programs can be computed in time polynomial in the dimension, number
of constraints and bits of precision (i.e. $\log 1/\eps$).

\subsection{Entanglement detection}\label{sec:witness}
So far, we have restricted ourselves to optimization problems over the
convex sets $\Sep$ and $\ProdSym$. In practice, another very important
problem is entanglement detection, i.e. testing whether a given
density matrix is a member of $\Sep$ or $\ProdSym$. In general,
membership testing and optimization for convex sets are intimately
related. There exist polynomial time reductions in both directions
using the ellipsoid method, as described in Chapter 4 of
\cite{gls:1993}. Thus, our results immediately imply an algorithm of complexity
$O(d^{\poly(n)} \poly\log(1/\epsilon))$ for membership testing in
$\Sep$. 

There is, however, a more direct way to go from optimization to
membership, using the notion of an \emph{entanglement witness}. The
idea is that to show that a given state $\rho$ is not in $\Sep$
(resp. $\ProdSym$), it suffices to find a Hermitian operator $Z$ such that
$\Tr[Z \rho] < 0$, but for all $\rho' \in \Sep$ (resp. $\ProdSym$),
$\Tr[Z \rho'] \geq 0$. Such an operator $Z$ is called an entanglement
witness for $\rho$. The search for an entanglement witness can be
phrased as an optimization problem:
\begin{equation}
  \label{eq:witness_opt}
  \begin{aligned}
    &\min_Z && \Tr[Z \rho] \\
    &\text{such that} && \Tr[Z \rho'] \geq 0 \; \forall \rho' \in \ProdSym
  \end{aligned}
\end{equation}
If the optimum value is less than $0$, then we know that $\rho$ is
entangled. Geometrically, an entanglement witness is a
separating hyperplane between $\rho$ and the convex set of separable
states. Thus, because of the hyperplane separation theorem for convex sets, every entangled $\rho$ must have some
witness that detects it. However, finding the witness may be very
difficult.

The witness optimization problem (\ref{eq:witness_opt}) is closely
related to the problem $h_{\ProdSym}$. In particular, suppose that for a
measurement operator $M$, we know that $h_{\ProdSym}(M) < \nu$. Then
$Z = \nu \ident - M$ is a feasible point
for (\ref{eq:witness_opt}). As a consequence of this, any feasible
solution to the SOS form of either DPS or our hierarchy will yield an
entanglement witness operator.

In the case of DPS, it turns out that this connection also yields an
efficient way to search for a witness detecting a given entangled
state. To see this, we consider the set of all possible witnesses
generated by DPS at level $r$, for any measurement operator $M$. Through straightforward computations (see Section VI of
\cite{dps:2003}), one finds that this set is
\be \label{eq:dps_witness} \mathrm{EW}_{\text{DPS}}(r) 
= \{ \Lambda^\dagger(Z_0 + Z_1 + \dots + Z_r)
: Z_0 \succeq 0, Z_1^{T_1} \succeq 0 \dots Z_r^{T_m} \succeq 0\}. \ee
Here $\Lambda$ is a certain fixed linear map and the superscripts $T_1, \dots T_m$
indicate various partial transposes (i.e. permutations interchanging a
subset of the row and column indices). The important thing to note is
that this is a convex set; in fact, it has the form of the feasible
set of a semidefinite program. Thus, given a state, it is possible to
efficiently search for an entanglement witness detecting it using a
semidefinite program.

Once we add the KKT conditions, the situation is not as
convenient. The set of all entanglement witnesses at level $r$,
denoted $\mathrm{EW}_{KKT}(r)$, is the set of $Z$ for which  $\exists
\sigma(x), \chi_{ij}(x)$ such that 
$$ \begin{aligned} &\<Z, x^{\ot 2d}\>
= \sigma(x) + \sum_{ij} \chi_{ij}(x) g_{ij}(x) \\ &\deg(\sigma(x)) \leq
r \\ &\deg(\chi_{ij}(x) g_{ij}(x)) \leq r \end{aligned}
$$
The important difference from DPS is that the polynomials $g_{ij}(x)$
come from the KKT conditions and thus \emph{depend on $Z$}. This in
particular means that $\mathrm{EW}_{KKT}(r)$ no longer has the form of
an SDP feasible set, nor indeed is it necessarily convex. However, we also note that by
Theorem~\ref{thm:sos_cert}, an open dense subset of all entanglement witnesses is contained in
$\mathrm{EW}_{KKT}(r)$ for $r = n^{O(d^2)}$.

\subsection{Numerical results}
\label{sec:numerics}
While our theoretical results show that adding the KKT conditions
results an improvement at very high levels of the DPS hierarchy,
we have also found numerical evidence of improvements even for
low-dimensional systems at very low levels of the
hierarchy. We compared the performance of the hierarchy with and
without the KKT conditions at the second level (i.e. searching over
SOS certificates of degree $6$) on a family of measurements with local
dimension $3$. The measurements were obtained by the applying the
construction in section VIII.A of~\cite{dps:2003} to the entanglement
witness given in equation~(69) of the same reference. Explicitly, they
are given by
\begin{equation} M_\gamma := \ident\otimes \ident - (A_\gamma^{-1} \otimes \ident) Z (A_\gamma^{-1} \otimes
\ident), \qquad \gamma \in [0, 1], \label{eq:dps_example} \end{equation}
where 
\begin{align*}
  A_\gamma &= \diag(1, \gamma, \dots, \gamma) \\
  Z &= 2(\ket{00}\bra{00} + \ket{11}\bra{11} + \ket{22}\bra{22}) \\
  &\qquad\qquad +
      \ket{02}\bra{02} + \ket{10}\bra{10} +\ket{21}\bra{21} -
      3\ket{\psi_+}\bra{\psi_+} \\
  \ket{\psi_+} &= \frac{1}{\sqrt{3}} \sum_{i=1}^{3} \ket{ii}.
\end{align*}
By construction, $h_{\Sep}(M_\gamma) \leq 1$ for all $\gamma \in
[0,1]$. However, it was shown in~\cite{dps:2003} that for sufficiently
small $\gamma$, the optimum value of DPS applied to $M_\gamma$ will be
strictly greater than $1$. Numerically, we find that this behavior
occurs for $\gamma < 0.1$. In Figure~\ref{fig:numerics}, we plot the
optimum value returned by the second level of the hierarchy for a
range of values of $\gamma$ between $0.01$ and $0.07$. We find that
adding the KKT conditions substantially improves the convergence. The
calculations were performed using YALMIP optimization
package~\cite{Lofberg2004,Lofberg2009}, the
SDP solver Mosek~\cite{mosek}, and the SDP preprocessing package frlib~\cite{Permenter14}.

\begin{figure}[h]
\centering
\subfloat[][]{
  \includegraphics[width=0.5\textwidth]{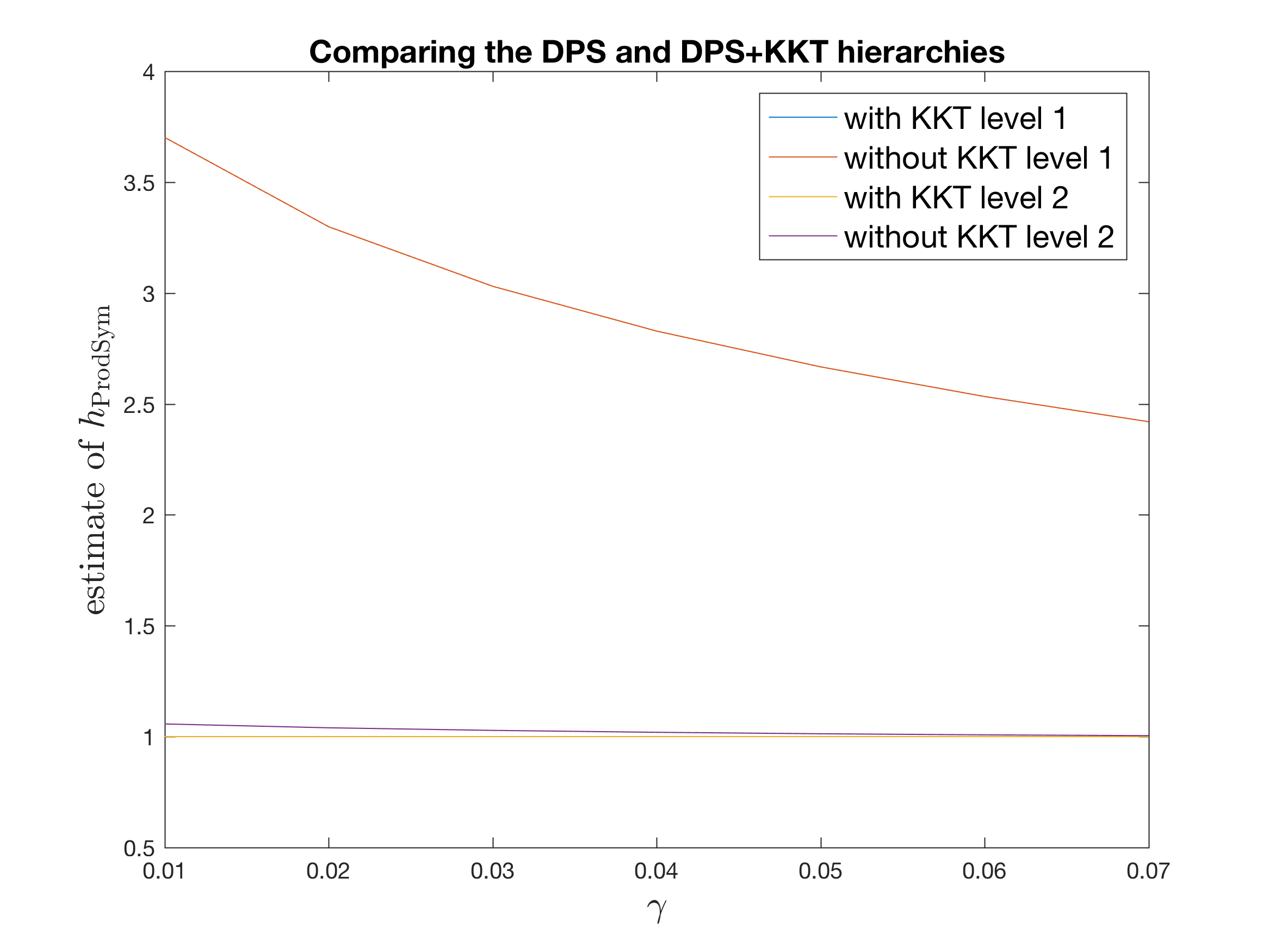}
}
\subfloat[][]{
  \includegraphics[width=0.5\textwidth]{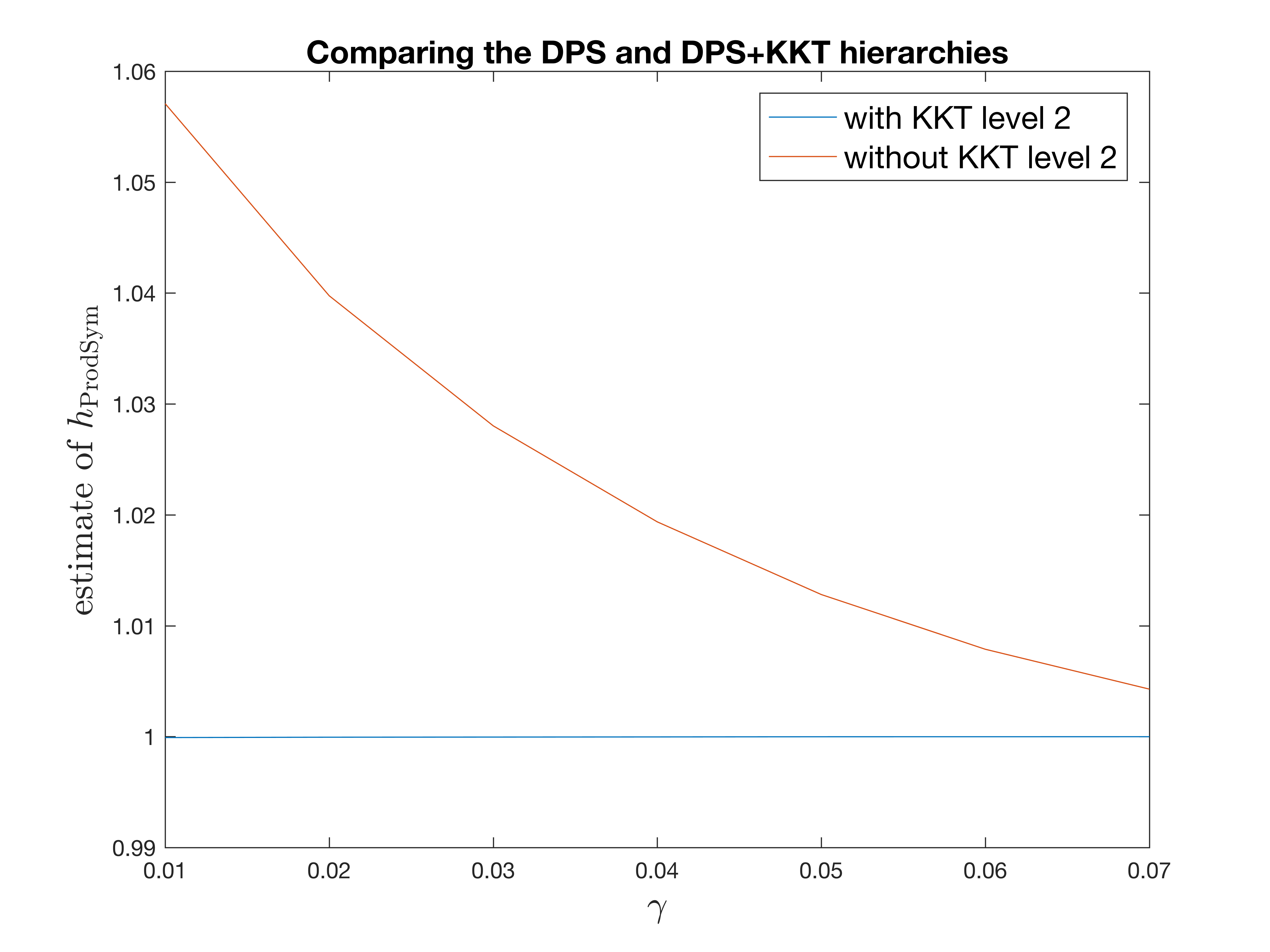}
}
\caption{Performance of hierarchy with and without KKT conditions, for
  the family of measurements $M_\gamma$ in Eq.~\eqref{eq:dps_example}. The true
  value of $h_{\Sep}$ is $\leq 1$ for all $\gamma$. Figure (a) shows
  the performance with and without KKT conditions at both level 1
  (searching over degree-4 SoS certificates) and level 2 (degree-6 SoS
certificates) for both variants of the hierarchy, while figure (b)
shows only level 2. At level 1, the KKT constraints have no effect and
both hierarchies yield the same relaxation. At level 2, the hierarchy
with KKT immediately converges to the true value $1$, within numerical
error, while the
hierarchy without KKT obtains upper bounds that are strictly greater
than $1$.}
\label{fig:numerics}
\end{figure}

\section{Proofs}\label{sec:proofs}

In this section, we will make use of a number of tools from algebraic
geometry, which are described in Appendix~\ref{sec:alg-geom}. At a high level, the proof will proceed as follows: first we show that for
generic $M$, the KKT ideal $I_K$ is zero-dimensional. This implies
that a Gr\"{o}bner basis of exponential degree can be found for
$I_K$. We then complete the proof using a strategy due to Laurent
(Theorem 6.15 of \cite{laurent:2009}): we start with a SOS certificate of
high degree, and then use division by the Gr\"{o}bner basis to reduce
the degree. This will result in a SOS certificate whose degree is the
same order as the degree of the Gr\"{o}bner basis, thus proving the theorem.

\subsection{Generic inputs}
We will now show that, for generic $M$, the KKT ideal is
zero-dimensional, using a dimension-counting argument based on the
theorems in Section~\ref{subsec:alggeo_intersection}. A similar result was
proved in Proposition 2.1 (iii) of \cite{nie:2009}. However, that result required both the
objective function and the constraints to be generic. Since the norm
constraint is fixed independent of the input $M$, this means we cannot
apply the result of \cite{nie:2009} directly. Nevertheless, we find that we can use a
very similar argument.

\begin{lemma}
  For generic $M$, the KKT ideal $I_K$ is zero dimensional. 
\end{lemma}

The intuition behind the proof is the same reason that the KKT
conditions characterize optimal solutions.  Roughly speaking the KKT
conditions encode the fact that at an optimal solution one should not
be able to increase the objective function without changing one or
more of the constraint equations.  This corresponds to a particular
Jacobian matrix having less than full rank.  Here we will see that
this rank condition on a Jacobian directly implies that the set of
solutions is zero dimensional.  

\begin{proof}
For the proof we will find it is useful to move to complex projective
space $\PP^n$,
parametrized by homogeneous coordinates $\tilde{x} = (x_0, x_1,
\dots, x_n)$.  For a polynomial $p(x)$, we denote its homogenization
by $\tilde{p}(\tilde{x})$. We also
define the following projective varieties
\begin{align*}
  \UU &= \{ \tilde x : \tilde{f}_1(\tilde{x}) = 0 \} \\
  \WW &= \{ \tilde x: \forall \, i, j, \tilde{g}_{ij}(\tilde x) = 0 \} \\
\end{align*}
The variety associated with the KKT ideal $V(I_K)$ is just the affine
part of $\UU \cap \WW$. So it suffices to show that $\UU \cap \WW$ is
finite. We will do this using a dimension counting
argument. Specifically, we will construct a variety of high dimension
that does not intersect $\WW$. By B\'{e}zout's Theorem, this will give us
an upper bound on the dimension of $\WW$.

To find such a variety, consider the
family $\HH$ of all hypersurfaces $\XX$ in $\PP^n$ of the form $\{ \tilde{f}_0(\tilde{x}) - \mu
x_0^{2d} = 0 \}$, parametrized by $\mu\in \bbC$ and the matrix
$M \in \bbC^{n^2 \times n^2}$. Multiplying
$\mu$ and $M$ by a nonzero scalar leaves the associated hypersurface
unchanged, so we can think of $(M, \mu)$ as a point in a projective
space $\PP^k$. We will be interested in the intersection $\AA = \XX
\cap \UU$ of a hypersurface $\XX$ in this family with the feasible set
$\UU$. The Jacobian
matrix $\tilde{J}_\AA$ of such an intersection is given by 
\[ \tilde{J}_\AA = \begin{pmatrix} \frac{\pd}{\pd x_0} (\tilde{f}_0(\tilde{x}) - \mu x_0^{2d}) & \frac{\pd
  \tilde{f}_1(\tilde{x})}{\pd x_0} \\
  \frac{\pd \tilde{f}_0(\tilde{x})}{\pd x_1} & \frac{ \pd \tilde{f}_1(\tilde{x})}{\pd x_1} \\
  \vdots & \vdots \\
  \frac{\pd \tilde{f}_0(\tilde{x})}{\pd x_n} & \frac{\pd \tilde{f}_1(\tilde{x})}{\pd x_n}
  \end{pmatrix}. \]
Let $J_\AA$
denote the submatrix of
$\tilde{J}_\AA$ obtained by removing the first row. We
claim that for a generic choice of $M$ and $\mu$, the matrix $J_\AA$
is of rank 2 everywhere on $\AA$. Since $\WW$ is the set of points
with $\rank \tilde{J}_\AA$, this implies that $\AA \cap \WW = \emptyset$.

Now, to prove the claim, we use Bertini's Theorem (Theorem~\ref{thm:bertini}). The variety $\UU$
is smooth and has dimension $n - 1$, and as long as $M\neq 0$, there are no points in common to all the
hypersurfaces in $\HH$. Thus, by Theorem \ref{thm:bertini}, for a
generic choice of $(M, \mu) \in \PP^k$, the
variety $\AA = \UU \cap \{ \tilde{f}_0(\tilde{x}) - \mu x_0^{2d} = 0 \}$ is
smooth (has no singular points) and has dimension $n-2$. 
This means that $\tilde{J}_\AA$ must have
rank $2$ everywhere on $\AA$. By homogeneity, we know that if
$\tilde{f}_0(\tilde{x}) - \mu x_0^{2d} = 0$, then
$\tilde{f}_0(\lambda \tilde{x})  - \mu (\lambda x_0)^{2d} = 0$ for all $\lambda \neq
0$. If we take the derivative of this expression with respect to
$\lambda$ and set $\lambda=1$, we get that $x_0 \frac{\pd}{\pd x_0} (\tilde{f}_0(\tilde{x}) -\mu x_0^{2d})= -\sum_i
x_i \frac{\pd}{\pd x_i} f_0(\tilde{x})$. Likewise we also find that $x_0 \frac{\pd}{\pd x_0} \tilde{f}_1(\tilde{x}) = -\sum_i
x_i \frac{\pd}{\pd x_i} f_1(\tilde{x})$. So whenever $x_0 \neq 0$, the first row
of $\tilde{J}_\AA$ is in the span of the other rows. 
Hence, for $x_0
\neq 0$, $\rank(\tilde{J}_\AA) = 2$ implies that $\rank(J_\AA) =2$ as
well. This means that that the affine part ($x_0 \neq 0$) of $\AA$ does not
intersect the affine part of $\WW$. It only remains to check the
part at infinity ($x_0 = 0$). We know that since $\AA$ is smooth,
$\tilde{J}_\AA$ has rank $2$ here also. By direct evaluation, we see that
the first row of $\tilde{J}_\AA$ is zero when $x_0 = 0$, so $J_\AA$
has rank $2$ here as well. Therefore, $\AA$ does not intersect $\WW$ anywhere.

Now we complete the proof using a dimension-counting
argument. B\'{e}zout's Theorem (Theorem~\ref{thm:bezout}) states that any two projective
varieties in $\PP^n$, the sum of whose dimensions is at least $n$,
must have a non-empty intersection. Thus, since $\WW \cap \AA = (\WW \cap \UU) \cap \{
\tilde{f}_0(\tilde{x}) = \mu x_0^{2d}\} = \emptyset$, we deduce that
\[ \dim(\WW \cap \UU) + \dim(\{
\tilde{f}_0(\tilde{x}) = \mu x_0^d\}) = \dim(\WW \cap \UU) + n - 1 <
n. \]
This implies that $\WW \cap \UU$ has dimension $0$, i.e. it is a
finite set of points in $\PP^n$. So $\WW \cap \UU \cap \{x_0 = 1\}$ is
a finite set of points in $\CC^n$. But this is precisely the variety
associated with the KKT ideal, or rather its complex
analogue. However, the fact that the KKT equations have a finite set
of solutions in $\CC^n$ implies that their set of solutions in
$\bbR^n$ is also finite.  Thus, the KKT ideal is zero-dimensional
as claimed.
 \end{proof}

For the next result, we will want to consider the ideal generated by a
homogenized version of the KKT conditions. For convenience sake, we
would like all the generators to be homogeneous of the \emph{same}
degree. The polynomials $g_{ij}(x)$ are already homogeneous and have degree
$2d$. The polynomial $f_1(x)$ is not homogeneous and has degree 2. So we
will homogenize it and multiply it by $x_0^{2(d-1)}$ to make it also degree
$2d$. This yields the following ideal
\[ \tilde{I}_K = \left\langle g_{ij}(\tilde{x}), x_0^{2(d-1)} \tilde{f_1}(x) \right\rangle. \]
\begin{lemma}
  The ideal $\tilde{I}_K$ has a Gr\"{o}bner basis in
  the degree ordering whose elements have degree
  $O(d^{\poly(n)})$. Moreover, each Gr\"{o}bner basis element $\gamma_k(\tilde{x})$
  can be expressed in terms of the original generators as $\gamma_k(\tilde{x}) =
  \sum_{ij} u_{ijk}(\tilde{x}) g_{ij}(\tilde{x}) + v_k(\tilde{x})(x_0^{2(d-1)} \tilde{f}_1(\tilde{x}))$ where
    $\deg(u_{ijk}(\tilde{x})), \deg(v_k(\tilde{x})) = O(d^{\poly(n)})$.
    \label{lem:kkt_groebner}
\end{lemma}
\begin{proof}
  Let $D$ be the degree of the Gr\"{o}bner basis.  Since the KKT ideal is zero dimensional, the homogenized KKT ideal
  is one-dimensional (that is, $V(\tilde{I}_K)$ is one-dimensional
  when viewed as an affine variety in $\CC^{n+1}$). So the result of
  Proposition~\ref{prop:groebner_degree} evaluated at $r = 1$ gives
  a bound $D = O(d^{n^2})$.
  Moreover, since
  the ideal is homogeneous, by Proposition~\ref{prop:groebner_homo} the Gr\"{o}bner basis elements can be
  chosen to be homogeneous as well.
  We will denote this Gr\"{o}bner basis of homogeneous polynomials as
  $\{ \tilde{\gamma}_k(\tilde{x}) \}$.

  Now, we know that any given Gr\"{o}bner basis element can be
  expressed in terms of the original generators from \eq{gij-def}:
  \[ \tilde{\gamma}_k(\tilde{x}) = \sum_{ij} u_{ijk}(\tilde{x}) \tilde{g}_{ij}(\tilde{x}) +  v_k(\tilde{x}) (x_0^{2d} \tilde{f}_1(\tilde{x})), \]
  where the polynomials $u_{ij}(\tilde{x})$ and $v_k(\tilde{x})$ could have arbitrarily high
  degree. Let the degree of $\tilde{\gamma}_k(\tilde{x})$ be $D_k \leq D$. Since it is
  homogeneous, all the terms on the RHS must be of degree
  $D_k$. Moreover, we know that $\tilde{g}_{ij}(\tilde{x})$ and $x_0^{2(d-1)}
  \tilde{f}_1(\tilde{x})$ are homogeneous of degree $2d$. Therefore, any terms in
  $u_{ijk}(\tilde{x})$ or $v_k(\tilde{x})$ with degree higher than $D_k$ will result only in terms
  of degree higher than $D_k+2d$ on the RHS. We know that these terms
  must cancel out to zero. Therefore, we can just drop all terms with
  degree higher than $D_k$ from $u_{ijk}(\tilde{x})$ and $v_k(\tilde{x})$ and equality will
  still hold in the equation above. Thus, we have shown that every
  Gr\"{o}bner basis element can be expressed in terms of the original
  generators with coefficients of degree at most $D$ as desired.
\end{proof}

Now we prove Theorem \ref{thm:sos_cert}. The argument is the same as
case (i) of Theorem 6.15 in \cite{laurent:2009}.

\begin{proof}
  Let $\{\tilde{\gamma}_i(\tilde x)\}$ be a degree-ordered Gr\"{o}bner basis for
  $\tilde{I}_{KKT},$ as in the previous proposition. By
  dehomogenizing, we get a Gr\"{o}bner basis $\{\gamma_i(x)\}$ for
  $I_k$. Since $1 - \sum_i x_i^2 \equiv 0 \pmod{I_{KKT}}$, the KKT ideal satisfies the
  Archimedean condition and Theorem \ref{thm:psatz} holds.
  Thus, there exists some $\sigma(x)$ SOS and
  $g(x) \in I_K$ such that $\nu - f_0(x) = \sigma(x) + g(x)$. Let us
  write $\sigma(x)$ explicitly as
  \[ \sigma(x) = \sum_a s_a(x)^2. \]
  Since $I_{KKT}$ is zero-dimensional, by Proposition~\ref{prop:groebner_zerod_division}, each term $s_a(x)$
  can be written as $s_a(x) = \sum a_{ak}(x) \gamma_k(x) + u_a(x)
  \equiv g_a(x) + u_a(x)$, where $\deg(u_a(x)) \leq nD$ and $g_a(x) \in
  I_{KKT}$. If we substitute this decomposition into the expression
  for $\sigma(x)$, we get
  \[ \sigma(x) = \sum_a u_a(x)^2 + g'(x), \]
  where $g'(x) \in I_{KKT}$. We can combine the terms in $I_{KKT}$ to
  get the following expression for the SOS certificate:
  \[\nu - f_0(x) = \sigma'(x) + g''(x), \]
  where $g''(x) \in I_{KKT}$ and $\deg(\sigma'(x)) \leq 2nD = d^{O(n^2)}$. Now, the LHS of this
  expression has degree $2d < \deg(\sigma'(x))$, so $g''(x)$ must also have degree
  $d^{O(n^2)}$. By Proposition~\ref{prop:groebner_division}, it can
  be expressed as
  \[ g''(x) = \sum h_{k}(x) \gamma_{k}(x), \]
  where $\deg(h_{k}(x) \gamma_{k}(x)) = d^{O(n^2)}$. Using
  Lemma~\ref{lem:kkt_groebner}, we can 
  express this in terms of the original generators as
  \[ g''(x) = \sum_{ijk} h_{k}(x) u_{ijk}(x) g_{ij}(x). \]
  We know that $\deg(u_{ijk}(x)) = d^{O(n^2)}$. Therefore, $g'(x) \in
  I_K^m$ for $m = d^{O(n^2)}$. This proves the theorem.
\end{proof}

\subsection{An algorithm for all inputs}
We have shown that for generic $M$, there exists a SOS certificate of
low degree for the optimization problem
\eqref{eq:hsep_kkt_sos}. However, for nongeneric $M$, it is possible
that no certificate of low degree exists, so the SOS formulation of
the hierarchy may not converge within $d^{O(n^2)}$ levels. In this section, we will
show that this problem goes away if we switch to the moment matrix
formulation \eqref{eq:hsep_kkt_moment} of the hierarchy. We will show that this formulation
converges in $d^{O(n^2)}$ levels for \emph{any} input $M$. First,
we show that the SDPs of the moment hierarchy are well behaved in the
sense that they satisfy \emph{Slater's condition} for any input
$M$. This is the condition that either the primal or dual feasible set
of the SDP should have a nonempty relative interior. To show this, we use
the following result from~\cite{trnovska:2005, josz:2014}.
\begin{proposition} 
  For a given SDP, let $\mathcal{P}, \mathcal{D},$ and $\mathcal{P}^*$ be the primal
  feasible set, dual feasible set, and set of primal optimal points,
  respectively. Then $\mathcal{P}$ and
  $\mathrm{interior}(\mathcal{D})$ are nonempty iff $\mathcal{P}^*$
  is nonempty and bounded.
\label{prop:slater}
\end{proposition}
In our case, let~\eqref{eq:hsep_kkt_moment} be the primal
and~\eqref{eq:hsep_kkt_sos} be the dual. The primal feasible set is
nonempty, since the true optimizing point for the unrelaxed problem
$h_{\ProdSym}$ is always feasible. Moreover, primal feasible set is
compact. Thus, the primal \emph{optimal} set $\mathcal{P}^*$ is nonempty
and bounded, and thus by Proposition~\ref{prop:slater}, Slater's
condition holds.

Slater's condition implies strong duality, so for generic $M$,
\eqref{eq:hsep_kkt_moment} and \eqref{eq:hsep_kkt_sos} give the same
optimum value. It also implies that the SDP value is a differentiable
function of the input parameters. We use this to extend our results to
non-generic inputs $M$.
\begin{theorem}
  For all input $M$, the hierarchy \eqref{eq:hsep_kkt_moment} converges to the optimum
  value of \eqref{eq:hprodsym} at level
  $r = d^{O(n^2)}$.
\end{theorem}
\begin{proof}
  For a given $M$, let $f^*_{\text{mom}, r}(M)$ be the optimum value of
  the $r$-th level of the hierarchy \eqref{eq:hsep_kkt_moment}. It is
  easy to see that $h_{\ProdSym}(M)$ is a continuous function of
  $M$~\footnote{One way to show this is to note that $h_{\ProdSym}(M)$ is a
    norm of $M$.}. We claim
  that $f^*_{\text{mom},r}(M)$ is also continuous. Indeed, Theorem 10 of
  \cite{shapiro:1997} states that if an SDP satisfies Slater's
  condition and has a nonempty bounded feasible set for all input
  parameters, then the optimum value is a differentiable function of
  the inputs. By the preceding discussion, these conditions hold for the moment hierarchy for all
  $M$, so $f^*_{\text{mom},r}(M)$ is indeed continuous.

  Now, by the remarks above, $h_{\ProdSym}(M) = f^*_{\text{mom},r}(M)$ for all generic
  $M$. Recall from Section~\ref{subsec:alg-geom-basics} that the set
  of generic $M$ is an open, dense set, according to the standard
  topology. Thus, since both functions $h_{\ProdSym}$ and $f^*_{\text{mom},r}$ are continuous and agree on an
  open dense subset, $h_{\ProdSym}(M)
  = f^*_{\text{mom},r}(M)$ for \emph{all} $M$. 
\end{proof}
\begin{corollary}
  For all input $M$, $h_{\ProdSym}(M)$ can be approximated up to
  additive error $\epsilon$ in time $O(d^{\poly(n)} \poly\log(1/\epsilon))$.
\end{corollary}

\section{Discussion and Open Questions}

Adding the KKT conditions provides a new way of sharpening the
familiar DPS hierarchy for testing separability.  We have given some
evidence that its asymptotic performance is superior to that of the
original DPS hierarchy.    Indeed, \cite{SO12} shows that even for
constant $n$, a variant of the $r^{\text{th}}$ DPS hierarchy has error
lower-bounded by $\Omega(1/r)$.  But our hierarchy converges in a
constant number of steps for any fixed local dimension. 

Does this mean that our hierarchy has other asymptotic improvements
over the DPS hierarchy at lower values of $r$?  We have seen already
cases in which DPS dramatically outperforms the weaker
$r$-extendability hierarchy.  For example, if $M$ is the projector
onto an $n$-dimensional maximally entangled state, then its maximum
overlap with PPT states is $1/n$ while its maximum overlap with
$r$-extendable states is $\geq 1/r$.  A more sophisticated example of
this scaling based on an $M$ arising from a Bell test related to the
unique games problem is in \cite{BuhrmanRSW11}.  One of the major open
questions in this area is whether low levels of SDP hierarchies such
as DPS can resolve hard optimizations problems of intermediate
complexity such as the unique games problem~\cite{BBHKSZ12}.

\section*{Acknowledgments}
AWH was funded by NSF grant CCF-1111382 and CCF-1452616. AN was
funded by a Clay Fellowship. AN also thanks Cyril Stark for helpful
conversations.  All three authors (AWH, AN and XW) were funded by ARO
contract W911NF-12-1-0486.

\appendix
\section{Algebraic Geometry} \label{sec:alg-geom}
\label{subsec:alg-geom-basics}

In this paper we will use some basic tools from algebraic
geometry, which we define in this section. The material presented here
can all be found in basic textbooks like~\cite{clo:1996,harris:1992}.

At the most basic level, algebraic geometry is about sets of zeros of
polynomial functions. Throughout this paper, we will be working with polynomials in $n$
complex variables $x_1, \dots, x_n$. We denote the ring of such polynomials by
$\CC[x_1, \dots, x_n]$. A fundamental concept in algebraic geometry is
the polynomial ideal:
\begin{definition}
The \emph{polynomial ideal} $I$ generated by polynomials $g_1(x), \dots,
g_k(x) \in \CC[x_1, \dots, x_n]$ is the set
\[ I = \left\{ \sum_{i=1}^k a_i(x) g_i(x) : a_i(x) \in \CC[x_1, \dots, x_n]
\right\}. \]
\end{definition}
The polynomials $g_i(x)$ are called a
\emph{generating set} for the ideal, and we write $I =\langle g_1(x),
\dots, g_k(x) \rangle$. Note that the same ideal can be generated by many
different generating sets.

Another fundamental concept is the algebraic variety:
\begin{definition}
A set $V \in \CC^n$ is called an \emph{(affine)
  algebraic variety} if $V = \{ x: u_1(x) = \dots = u_k(x) = 0\}$ for
some polynomials $u_1(x), \dots, u_k(x)$.
\end{definition}

Every ideal $I$ has an associated
variety $V(I)$, which is the set of common zeros of all polynomials in
$I$ (or equivalently, the set of common zeros of all the generators of
$I$ for any generating set).

In this paper, we will be using some theorems concerning intersections of
varieties. These properties are most conveniently stated not in
$\CC^n$, but in the complex \emph{projective} space $\PP^n$. There are
several ways to define $\PP^n$, but for our purposes it will be most
convenient to use \emph{homogeneous coordinates}: we define $\PP^n$ as the
set of all points $(x_0, x_1, \dots, x_n) \in
\CC^{n+1} - \{0\}$ up to multiplication by a nonzero constant. Thus,
$(x_0, x_1, \dots, x_n)$ denotes the same point as $(\lambda x_0,
\lambda x_1, \dots, \lambda x_n)$. Henceforth, we will denote the
homogeneous coordinates using $\tilde{x}$. The hyperplane $x_0 = 0$
can be thought of as the set of ``points at infinity.'' 

We define a \emph{homogeneous polynomial} to be the sum of monomial
terms that are all of the same degree. Given any polynomial function
$f(x)$ on $\CC^n$ of degree $d$, we define its homogenization by
$\tilde{f}(\tilde{x}) = x_0^d f(x_1/x_0, \dots, x_n/x_0)$. Using these
concepts, we can define a
\emph{projective algebraic variety} as a set of the form $\VV = \{ \tilde{x} \in
\PP^n : \tilde{u}_1(\tilde{x}) = \dots = \tilde{u}_k(\tilde{x}) =
0\}$, where $\tilde{u}_i(\tilde{x})$ are homogeneous polynomials. Given any affine variety in $\CC^n$, we can
produce a corresponding projective variety on $\PP^n$ by homogenizing the
defining polynomials. Likewise, we can go from a projective variety to
an affine variety by dehomogenizing, i.e.~intersecting with $\{x_0 =
1\}$. 

In general, an algebraic variety may not be a smooth manifold in
$\CC^n$ or $\PP^n$---it may have one or more singular points. A
criterion for smoothness can be obtained from the Jacobian matrix
associated with the variety. The Jacobian matrix of the variety $V =
\{x \in \CC^n: u_1(x) = \dots = u_k(x) = 0\}$ is given by
\[ J = \begin{pmatrix} \frac{\pd u_1(x)}{\pd x_1} & \dots & \frac{\pd
    u_k(x)}{\pd x_1} \\
  \vdots & \ddots & \vdots \\
   \frac{\pd u_1(x)}{\pd x_n} & \dots & \frac{\pd
    u_k(x)}{\pd x_n} \end{pmatrix}. \]
A point $x \in V$ is a singular point if the matrix $J$ has less than
full rank at $x$. $V$ is smooth if it has no singular points. The
\emph{codimension} of $V$ (i.e. $n-\dim V$) is equal to the rank of $J$ at nonsingular
points. This also coincides with the intuitive meaning of dimension (from differential geometry) as
applied to manifolds. If a variety on $\CC^n$ or $\PP^n$ has dimension $n-1$, we
call it a \emph{hypersurface}. Using the correspondence between ideals
and varieties, we can also define the dimension of an ideal $I$ as the
dimension of the associated affine variety $V(I)$.

The last basic notion we will need is the idea of ``genericity.'' To
define this precisely in the context of algebraic geometry, we need to
introduce the Zariski topology. This is the topology over $\CC^n$ or
$\PP^n$ in which the closed sets are precisely the algebraic
varieties. We say that a property over points in $\CC^n$ or $\PP^n$ is
\emph{generic} if it is true for a Zariski open dense subset of $\CC^n$. Note
that all Zariski closed sets are also closed in the standard
topology, and therefore all Zariski open sets are open in the standard
topology. So if a set is generic in the sense defined here, it is also
open and dense in $\CC^n$ under the standard topology.

\subsection{Gr\"{o}bner bases}
We noted above that a polynomial ideal can have many different
generating sets. However, there is a notion of a canonical generating
set, called a \emph{Gr\"{o}bner basis}, that is computationally
useful. To define it, we must first define the notion of a monomial
ordering.
\begin{definition}  
A monomial ordering is any total ordering $\prec$ on the set of
monomials satisfying the following:
\begin{enumerate}[(i)]
\item If $a \prec b$, then for any monomial $c$, $ac \prec bc$.
\item Any nonempty subset of monomials has a smallest element (the
  well-ordering property).
\end{enumerate}
\end{definition}
An important class of monomial orderings is the degree
orderings: these are the orderings in which if $deg(a) >
deg(b)$, then $a \succ b$.

Once we have chosen a monomial ordering, for any polynomial $f(x)$ we can define the
\emph{leading term} $\LT(f(x))$ as the monomial term in $f(x)$ that is
highest according to our chosen ordering. With these notions in place,
we can define the Gr\"{o}bner basis as follows.
\begin{definition}
  A collection of polynomials $\{g_1(x), \dots, g_k(x)\}$ is a
  \emph{Gr\"{o}bner basis} of an ideal $I$ if $I = \langle g_1(x), \dots,
  g_k(x) \rangle$ and 
  \[ \langle \LT(g_1(x)), \dots, \LT(g_k(x)) \rangle = \langle \{ \LT(f(x))
  : f(x) \in I \} \rangle. \]
  \label{def:groebner}
\end{definition}
Gr\"{o}bner bases were introduced by Buchberger~\cite{buchberger:70}, who showed that every ideal has a finite Gr\"{o}bner
basis, and gave an algorithm to compute this basis for any given
monomial ordering.

A key application of the Gr\"{o}bner basis is 
in the \emph{Gr\"{o}bner basis division algorithm}. The output of this
algorithm is described in the following proposition.
\begin{proposition}
  Let $f(x)$ be any polynomial, and $I$ be an ideal with a
  degree-ordered Gr\"{o}bner basis $\{g_1(x), \dots, g_k(x)\}$. If $D$ is the
  maximum degree of the Gr\"{o}bner basis elements, then there exists
  a unique decomposition $f(x) = \sum
  a_i(x) g_i(x) + u(x)$, where $\deg(a_i(x)) \leq
  \deg(f(x))$, and no term of $u(x)$ is divisible by the leading term
  of a Gr\"{o}bner basis element. Moreover, if $f(x) \in I$, then $u(x) = 0$.
  \label{prop:groebner_division}
\end{proposition}
\begin{proof}
  This is an immediate consequence of Proposition 1
  in Section 2.6 and Theorem 3 in Section 2.3 of~\cite{clo:1996}. 
\end{proof}

If an ideal is generated by homogeneous polynomials, then the
degree-ordered Groebner basis can also be taken to be homogeneous.
\begin{proposition}
  Let $\tilde{I} = \langle \tilde{h}_1(\tilde{x}), \dots \tilde{h}_k(\tilde{x}) \rangle$ be an ideal
  generated by homogeneous polynomials, and $\{g_1(\tilde{x}), \dots, g_k(\tilde{x})\}$ be a
  degree-ordered Gr\"{o}bner basis for $\tilde{I}$. If we let $g'_i(\tilde{x})$ be the
  highest-degree terms of $g_i(\tilde{x})$, then $\{g'_1(\tilde{x}),
  \dots, g'_k(\tilde{x})\}$ is also a degree-ordered Gr\"{o}bner basis for
  $I$.
  \label{prop:groebner_homo}
\end{proposition}
\begin{proof}
  We need to show that $\{g'_1(\tilde{x}), \dots, g'_k(\tilde{x})\}$ is a
  generating set for $\tilde{I}$, and that the condition in
  Definition~\ref{def:groebner} still holds. The latter follows
  immediately from the fact that $\LT(g'_i(\tilde{x})) = \LT(g_i(\tilde{x}))$ for degree
  orderings. As for the former, suppose
  that $f(\tilde{x}) \in \tilde{I}$, meaning that $f(\tilde{x}) = \sum_{i}
  u_i(\tilde{x}) \tilde{h}_i(\tilde{x})$. Let $P_df$ denote the degree-$d$ terms of
  $f(\tilde{x})$. Then
  $P_d f(\tilde{x}) = \sum_i (P_{d - \deg(\tilde{h}_i)}u_i(\tilde{x}) ) \tilde{h}_i(\tilde{x})$, so
  $P_d f (\tilde{x}) \in \tilde{I}$. Now, for any Gr\"{o}bner basis element $g_i(\tilde{x})$,
  let $d < \deg(g_i(\tilde{x}))$. Since $P_d g_i(\tilde{x}) \in \tilde{I}$, by
  Proposition~\ref{prop:groebner_division}, $P_d g_i(\tilde{x}) = \sum a_{ij}(\tilde{x})
  g_j(\tilde{x})$, where the sum only contains Gr\"{o}bner basis elements with
  degree at most $d$. Since $d < \deg(g_i(\tilde{x}))$, this means in particular
  that this sum does \emph{not} include $g_i(\tilde{x})$. This implies that we
  can replace $g_i(\tilde{x})$ by $g_i(\tilde{x})
  - P_d g_i(\tilde{x})$, and still have a generating set for $\tilde{I}$. By repeatedly
  applying this process, we can replace each $g_i(\tilde{x})$ by $g'_i(\tilde{x})$ and still
  have a generating set. Thus, $\{g'_1(\tilde{x}), \dots, g'_k(\tilde{x})\}$ is indeed a
  Gr\"{o}bner basis for $\tilde{I}$.
\end{proof}

The dimension of an ideal is related to properties of its Gr\"{o}bner
basis. For ideals of any dimension, the following bound on the degree
of the Gr\"{o}bner basis was shown in~\cite{mayr:2010}.
\begin{proposition}
  For an $r$-dimensional
  ideal generated by polynomials of degree at most $d$ in $n$
  variables, with coefficients over any field, the Gr\"{o}bner basis in any ordering has degree
  upper-bounded by
  \[ 2\left(\frac{1}{2} d^{n-r} + d\right)^{2^r} . \]
  \label{prop:groebner_degree}
\end{proposition}
In the special case of zero-dimensional ideals, we further have the
following property:
\begin{proposition}
  Let $I$ be an ideal and $\{g_1(x), \dots, g_k(x)\}$ a Gr\"{o}bner basis
  for $I$. Then $I$ is zero-dimensional iff for every variable $x_i$, there
  exists $m_i \geq 0$ such that $x_i^{m_i} = \LT(g(x))$ for some
  element $g(x)$ in the Gr\"{o}bner basis. 
\end{proposition}
\begin{proof}
  This is the equivalence (i) $\iff$ (iii) in Theorem 6 of Chapter 5
  of~\cite{clo:1996}.
\end{proof}
This result enables us to bound the degree of the remainder term $u$
in Proposition~\ref{prop:groebner_division} above, when the ideal is
zero dimensional.
\begin{proposition}
  If $I$ is a zero-dimensional ideal over $n$ variables, and it has a
  degree-order Gr\"{o}bner basis whose
  maximum total degree is $D$, then the remainder $u(x)$ in
  Proposition~\ref{prop:groebner_division} has degree at most
  $n(D-1)$.
  \label{prop:groebner_zerod_division}
\end{proposition}
\begin{proof}
  Suppose $u(x)$ contains a term with degree greater than
  $n(D-1)$. Then this term would be divisble by $x_i^D$ for some
  variable $x_i$. However, since $I$ is zero dimensional, by the above
  proposition there exists a Gr\"{o}bner basis element $g_j$ whose
  leading term is $x_j^k$ for some $k < D$. Thus, we have found a term
  in $u(x)$ that is divisble by the leading term of a Gr\"{o}bner
  basis element, which contradicts
  Proposition~\ref{prop:groebner_division}.
\end{proof}
\subsection{Intersections of varieties}\label{subsec:alggeo_intersection}
Finally, we include two important theorems concerning the intersections
of projective algebraic varieties. In full generality these theorems
are much more powerful than we need; the statements we give here are
tailored for our use, and are based on those in~\cite{nie:2009}. The first theorem is B\'{e}zout's Theorem, 
which says that two projective varities of sufficiently high dimension
must intersect (the full version also bounds the number of components
in the intersection):
\begin{theorem}[B\'{e}zout]
  Suppose $\UU$ and $\VV$ are projective varieties in $\PP^n$, and
  $\dim(\UU) + \dim(\VV) \geq n$. Then $\UU$ and $\VV$ have a nonempty
  intersection.
  \label{thm:bezout}
\end{theorem}
The second theorem is Bertini's Theorem. Roughly, this states that the intersection
of a smooth variety with a ``generic'' hypersurface is also a smooth
variety with dimension $1$ lower. The precise statement is:
\begin{theorem}[Bertini]
  Let $\UU$ be a $k$-dimensional smooth projective variety in $\PP^n$, and $\HH$ a
  family of hypersurfaces in $\PP^n$ parametrized by coordinates in a projective
  space $\PP^m$. If there are no points common to all the
  hypersurfaces in $\HH$, then for generic $\AA \in \HH$, the intersection
  $\UU \cap \AA$ is smooth and has dimension $k - 1$.
  \label{thm:bertini}
\end{theorem}

\end{document}